\documentclass[11pt]{article}
\usepackage[margin=1in]{geometry}
\usepackage{amsfonts}
\usepackage{amsmath,amsthm,amssymb}
\usepackage{latexsym}
\usepackage{epic}
\usepackage{epsfig}
\usepackage{hyperref}
\usepackage{verbatim}
\usepackage{enumitem}
\usepackage[justification=centering]{caption}
\usepackage{enumitem}
\usepackage{color}
\allowdisplaybreaks[1]

\newtheorem{theorem}{Theorem}[section]
\newtheorem{corollary}[theorem]{Corollary}
\newtheorem{lemma}[theorem]{Lemma}

\newtheorem{claim}[theorem]{Claim}
\newtheorem{definition}[theorem]{Definition}

\newtheorem*{conjecture*}{Conjecture}
\newtheoremstyle{nonindented}{1ex}{1ex}{}{}{\bfseries}{.}{.5em}{}
\newtheoremstyle{indented}{1ex}{1ex}{\itshape\addtolength{\leftskip}{0.6cm}\addtolength{\rightskip}{0.6cm}}{}{\bfseries}{.}{.5em}{}
\theoremstyle{nonindented}
\theoremstyle{indented}
\theoremstyle{plain}

\newcommand{\cross}{\times}

\newcommand{\set}[1]{\left\{ #1 \right\}}
\newcommand{\union}{\cup}
\newcommand{\intersect}{\cap}

\renewcommand{\tilde}{\widetilde}

\DeclareMathOperator{\poly}{poly}

\def\max{\qopname\relax n{max}}

\def\argmax{\qopname\relax n{argmax}}

\def\Pr{\qopname\relax n{\mathbf{Pr}}}
\def\Ex{\qopname\relax n{\mathbf{E}}}

\newcommand{\RR}{\mathbb{R}}
\newcommand{\RRp}{\RR_+}

\def\A{\mathcal{A}}

\def\D{\mathcal{D}}

\def\F{\mathcal{F}}

\def\N{\mathcal{N}}
\def\P{\mathcal{P}}

\def\V{\mathcal{V}}

\def\eps{\epsilon}
\def\sse{\subseteq}

\newcommand{\eat}[1]{}

\newcommand{\maxi}[1]{\mbox{maximize} & {#1 } & \\}
\newcommand{\st}{\mbox{subject to} }
\newcommand{\con}[1]{&#1 & \\}
\newcommand{\qcon}[2]{&#1, & \mbox{for } #2.  \\}
\newenvironment{lp}{\begin{equation}  \begin{array}{lll}}{\end{array}\end{equation}}
\newenvironment{lp*}{\begin{equation*}  \begin{array}{lll}}{\end{array}\end{equation*}}

\author{Yu Cheng\footnotemark[3] \and Ho Yee Cheung\thanks{Supported in part by USC Provost's Ph.D. Fellowship and NSF CCF-0964481.} \and Shaddin Dughmi\thanks{Supported in part by NSF CAREER Award CCF-1350900.} \and Shang-Hua Teng\thanks{Supported in part by NSF CCF-0964481 and NSF CCF-1111270.}\\
Department of Computer Science\\ University of Southern California}
\title{Signaling in Quasipolynomial Time}
\date{}

\begin{document}

\maketitle
\thispagestyle{empty}
\begin{abstract}
Strategic interactions often take place in an environment rife with
  uncertainty. 
As a result, the equilibrium of a game is intimately
  related to the information available to its players. 
The \emph{signaling problem} abstracts the task faced by an informed
  ``market maker'', who must choose how to reveal information in order
   to effect a desirable equilibrium.

In this paper, we consider two fundamental signaling problems: 
  one for abstract normal form games, and the other for single item
  auctions. 
For the former, we consider an abstract class of objective
  functions which includes the social welfare and weighted combinations
  of players' utilities, and for the latter we restrict our attention to
  the social welfare objective and to signaling schemes which are
  constrained in the number of signals used. 
For both problems, we design approximation algorithms for
  the signaling problem which run in quasi-polynomial time under various
  conditions, extending and complementing the results of various recent
  works on the topic.

Underlying each of our results is a ``meshing scheme'' which
   effectively overcomes the ``curse of dimensionality''
   and discretizes the space of ``essentially different''
   posterior beliefs -- in the sense of inducing ``essentially
  different'' equilibria. 
This is combined with an algorithm for
  optimally assembling a signaling scheme as a convex combination of
  such beliefs. 
For the normal form game setting, the meshing scheme leads to a convex
  partition of the space of posterior beliefs and this assembly procedure
  is reduced to a linear program, and in the auction setting the
  assembly procedure is reduced to submodular function maximization.
\end{abstract}

\newpage
\setcounter{page}{1}
\section{Introduction}

\def\xx{\mathbf{x}}
\def\xxappr{\tilde{\mathbf{\xx}}}

\def\calG{\mathcal{G}}
\def\Nash{\mathrm{NashEquilibrium}}

In games with uncertainty, the
  information available to players 
  may influence their strategic decisions,
  and  could fundamentally impact their utilities and 
  their beliefs regarding the preferences of other players.
Consequently, a principal (market maker) who is privy to confidential state-of-nature information may affect the equilibrium 
  by strategically ``leaking'' information ahead of players'
  decisions --- a practice often referred to as \emph{signaling}.
If the principal's goal is to boost her own objective 
 in the outcome of the
  game, then her strategic decision regarding
 which information to reveal is an optimization problem,
 which we refer to as the {\em signaling problem}.

We consider {\em Bayesian games} 
  in which  players have {\em prior} beliefs
  on the structure of the game,  and 
  form  {\em posterior beliefs}  based on  
  the information revealed by the principal. We restrict our attention to games in which players share the same prior belief, and \emph{symmetric signaling schemes} which reveal the same information to all players.
Thus, a signaling scheme 
  defines a number of new Bayesian (sub-)games --- one for each realization of the information string (aka the \emph{signal}).
From this viewpoint, the search space for signaling
  is the space of decompositions of the prior belief into posterior beliefs.
The goal of the principal is then to identify a signaling scheme which induces, on average over the resulting Bayesian subgames, a favorable equilibrium.

We focus on two fundamental models previously considered in a signaling context. Our first model, considered in \cite{D14}, is the most abstract: normal form games in which the payoff matrix is parametrized by a state of nature, which is drawn from a common prior. We consider the signaling question in normal form games for an  abstract class of objectives which include the social welfare, weighted combinations of players' utilities, and utility functions which depend on players' actions in the various states of nature. Our second model, considered in \cite{emeksignaling,miltersensignals,DIR14} is that of a Bayesian second-price auction, in particular when there is uncertainty regarding the identity or attributes of the item for sale, and players' private valuations for the various configurations of the item are drawn from a distribution. We restrict attention to the social welfare objective in this model, yet as in \cite{DIR14} we place a bound on the communication bandwidth of the signaling scheme --- equivalently the number of different signals used to describe the item.

\subsection{Context and Results}

  There are significant computational challenges
  involved in obtaining an (approximately) optimal signaling scheme. Whereas the space of the posterior beliefs is a convex set, signaling problems in even very simple settings involve a non-convex objective, or  combinatorial constraints, which either lead to hardness results or require novel algorithmic approaches. This is borne out in a number of recent works which examine the algorithmic aspects of signaling, and we mention those most relevant to this paper next.

Signaling in the Bayesian second-price auction model has perhaps received the most attention from an algorithmic perspective. Emek et al.~\cite{emeksignaling} and Miltersen and Sheffet \cite{miltersensignals}, consider the signaling problem faced by an informed auctioneer who must decide to (partially) reveal information regarding the item for sale, with the goal of maximizing his revenue. Whereas the general problem is shown to be NP-hard, the special case in which players' valuation distribution has small support is solved in polynomial time. We also mention the work of Guo and Deligkas \cite{guo}, who consider a similar question albeit in a different model of uncertainty, and similarly present NP-hardness  results and heuristics for their model. 

More recently, Dughmi et al.~\cite{DIR14} consider the same auction setting, albeit with the welfare objective and subject to exogenous constraints on the signaling policy --- most notably a communication constraint. There the problem was shown to be NP-hard to approximate better than a $(1-1/e)$ factor, and approximation algorithms matching this guarantee were presented in the special case in which players' valuations are drawn from a small support distribution.  The present paper considers the general algorithmic question without  a bound on players' valuation distribution, while relaxing the runtime requirement to quasi-polynomial. Specifically, we derive a quasi-polynomial-time algorithm with an approximation ratio of $(1-1/e)$, modulo an additive loss of $\epsilon$ times the upper limit of a player's valuation. Our result holds when players' valuation distribution is given explicitly, and more generally whenever said distribution can be sampled efficiently.

The other signaling model considered in this paper, that of abstract normal form games, has only recently been examined from a complexity-theoretic perspective. Dughmi~\cite{D14} considered the special case of two-player zero-sum games, and examined the design of symmetric signaling schemes with the goal of maximizing the expected utility of one of the players. The results were mostly negative, assuming the conjectured hardness of the planted clique problem, Specifically, it was shown that no fully polynomial time approximation scheme (FPTAS) is possible for the signaling problem for zero sum games, assuming the planted clique conjecture. This leaves open the potential for a PTAS. In the present paper, we complement the impossibility result of \cite{D14} with a quasi-polynomial time approximation scheme (QPTAS), which applies more generally to general sum normal form games with a constant number of players, and an abstract class of objectives which includes the social welfare and weighted combinations of players' utilities as a special case. For general sum games, our result is a bicriteria result which forgoes the exact Nash equilibrium, and instead computes a signaling scheme with $\epsilon$-equilibria that are competitive with the $\epsilon'$-equilibria of every other signaling scheme whenever $\epsilon' < \epsilon$ (including $\epsilon' = 0$). 
Our result can also be extended to Stackelberg games.



\subsection{Our Techniques}

Both our results rely on a  ``meshing method'' applied to the space of posterior beliefs, or equivalently the equilibria which they induce. The method proceeds in two main steps.

\begin{enumerate}
\item {\em Construct a Small Dictionary of Equilibria}:  
This is a discrete family of objects which indexes the potential equilibria of a signaling scheme, with the property that they form an ``$\epsilon$-net'' of the space of all equilibria in a precise technical sense. This ``net'' forms our signaling ``dictionary''.
\item {\em Construct a near-optimal Signaling Scheme}:
 We then  assemble a near-optimal signaling scheme which induces subgames with equilibria from our dictionary.  This involves solving a nontrivial optimization problem which optimally decomposes the prior distribution into posterior beliefs inducing equilibria in our dictionary.
\end{enumerate}
The technical challenge involved in step (1) is the identification of a net of the space of equilibria which has ``low approximation dimensionality'' with respect to the space of signaling schemes and the design objective.
This arises due to the ``curse of dimensionality'' of the 
  space of posterior beliefs and the space of equilibria.  For step (2), the challenge arises from the fact that not every convex combination of posterior beliefs is a valid signaling scheme --- indeed, the signaling scheme must be a convex decomposition of the prior distribution. This induces a nontrivially constrained optimization problem.

For the results of Section \ref{sec:game} for normal form games, our dictionary is based on a Lemma of Lipton et al.~\cite{lmm03}. Specifically, their result implies the existence of a quasi-polynomial-sized family of mixed strategy profiles which, simultaneously for all games and equilibria of those games, includes a profile which approximates the payoffs of the equilibrium within an additive $\epsilon$, and itself forms an  $\epsilon$-equilibrium. To combine these approximate equilibria into a signaling scheme, we make two observations: First, the space of posterior beliefs inducing a particular equilibrium forms a convex polytope; Second, the optimization problem of optimally partitioning the prior belief into a quasi-polynomial number of posterior beliefs, one in each polytope corresponding to an equilibrium, can be formulated via a linear program after an appropriate change of variables.

For the results of Section \ref{sec:auction} for second-price auctions, constructing our dictionary takes more work. Specifically, we consider all rules for selecting a winner of the auction as a function of the drawn valuation matrix, and show that a quasi-polynomial number of such rules suffices to approximate the welfare in all equilibria to within an additive error of $\epsilon$. Assembling a signaling scheme which induces such winner selection rules in an optimal combination is then reduced to submodular function maximization subject to a cardinality constraint, which admits a $(1-1/e)$-approximation algorithm.

\subsection{Additional Discussion of Related Work}

The study of the effects of information on strategic interactions, and mechanisms for signaling, has its roots in the early works of Akerlof~\cite{akerloflemons} and Spence~\cite{spence1973job}. 
%
Hirshleifer~\cite{hirshleifer71} was the first to observe that optimal information revelation is nontrivial, in the sense that more information sometimes leads to worse market outcomes. This contrasts with earlier work by  Blackwell~\cite{blackwell51} which implied that more information is always better for a single agent in a non-competitive environment. Since then, many works have examined the effects of additional  information on players' equilibrium utilities.  Lehrer et al.~\cite{lehrermediation} showed that additional information improves players' utilities in \emph{common interest} games --- games where players have identical payoffs in each outcome, and  Bassan et al.~\cite{bassanpositive} exhibited a polyhedral characterization of games in which more information improves the individual utility of \emph{every} player.
The ``linkage principle'' of Milgrom and Weber~\cite{MW82} exhibits natural conditions under which revealing additional information to buyers in an auction improves the seller's revenue; and Syrgkanis et al.~\cite{CommonValue} examine failings of the linkage principle in common value auctions. Unlike revenue, it is known that additional information always improves welfare in a second-price auction.



Despite appreciation of the importance of information in strategic interactions, it is only recently that researchers have started viewing the information structure of a game as a mathematical object to be designed, rather than merely an exogenous variable. Kamenica and Gentzkow~\cite{bayesianpersuasion} examine settings in which a sender must design a signaling scheme to convince a less informed receiver to take a desired action. Recent work in the CS community, including by Emek et al.~\cite{emeksignaling}, Miltersen and Sheffet~\cite{miltersensignals}, and Guo et al.~\cite{guo}, examines revenue-optimal signaling in an auction setting, and presents polynomial-time algorithms and hardness results for computing it.  Dughmi et al.~\cite{DIR14} examine welfare-optimal signaling in an auction setting under exogenous constraints, and presents polynomial-time algorithms and hardness results under assumptions on players' value distributions; our auction signaling result applies to a generalization of their model. Also relevant to our work is the recent result of Dughmi~\cite{D14}, which essentially rules out an FPTAS for the signaling problem in zero-sum games; our result on general games is a positive counter-point to the result of \cite{D14}, slightly relaxing both the equilibrium definition and the polynomial-time restriction.

\section{Preliminaries: Information Revelation and Signaling Schemes}
\label{sec:prelims}
\subsection{Games}
\label{sec:games}

A \emph{Bayesian game} is a family of 
  \emph{games of complete formation} parametrized by a state of nature $\theta$,
  where $\theta$ is assumed to be drawn from a known prior
  distribution. 
We consider two classes of such games: \emph{explicit  normal form Bayesian games},
   and \emph{Bayesian single item auction games}.

\subsubsection*{Normal Form games}

We consider normal form games of incomplete information, 
  given by the following parameters.
\begin{itemize}
\item A positive integer $n$ denoting the number of players,
  which we index by the set $[n]=\set{1,\ldots,n}$.
\item A nonnegative integer $m$ bounding the number of pure strategies 
  of each player.\footnote{Without loss of generality, we assume each 
  player has exactly $m$ pure strategies.} 
\item A finite family $\Theta=\set{1,\ldots,M}$ of \emph{states of nature}, 
  which we index by $\theta$.
\item A family of \emph{payoff tensors} $\A_i^\theta: [m]^n \to [-1,1]$, 
  one per player $i$ and state of nature $\theta$, 
  where $\A_i^\theta(s_1,\ldots,s_n)$ is the payoff to 
  player $i$ when the state of nature is $\theta$ and 
  each player $j$ plays strategy $s_j$.
\end{itemize}

A \emph{Bayesian} normal form game of incomplete information is
  additionally equipped with a \emph{common prior distribution}
  $\lambda \in \Delta_{M}$ on the states of nature. 
Absent additional informational, risk neutral players behave as in the
  complete information game $\Ex_{\theta \sim \lambda} [\A^\theta]$. 
We consider signaling schemes which partially and symmetrically inform
  players by publicly announcing a signal $\sigma$, correlated with
  $\theta$; this induces a common posterior belief on the state of
  nature for each value of $\sigma$. 
When players' posterior belief over
  $\theta$ is given by $\mu \in \Delta_M$, we use $\A^\mu$ to denote the
  equivalent complete information game $\Ex_{\theta \sim \mu}
  [\A^\theta]$.  
As shorthand, we use $\A_i^\mu(x_1,\ldots,x_n)$ to
  denote $\Ex [\A_i^\theta(s_1,\ldots,s_n)]$ when $\theta \sim \mu \in
  \Delta_M$ and $s_i \sim x_i \in \Delta_m$. 
In the event that the state
  of nature is $\theta$ and players play the pure strategy profile
  $s_1,\ldots,s_n$, we refer to the tuple $(\theta,s_1,\ldots,s_n)$ as
  the \emph{state of play}.

For all our results, we assume that a Bayesian game 
  $(\A,\lambda)$ is represented explicitly as a 
  list of tensors 
  $\set{\A_i^\theta \in [-1,1]^{m^n} : i \in [n], \theta \in [M]}$, 
and a vector $\lambda \in \Delta_M$.


\subsubsection*{Single Item Second-Price Auctions}
We consider Bayesian second-price auctions described by the following parameters:

\begin{itemize}
\item A nonnegative integer $n$ denoting the number of bidders. 
We index the bidders by the set $[n]=\set{1,\ldots,n}$.
\item A finite family $\Theta=\set{1,\ldots,M}$ of \emph{states of nature}, 
  indexed by $\theta$. 
The states of nature  represent potential configurations of the item for sale. 
\item A common-knowledge prior distribution $\lambda \in \Delta_M$ on the 
   states of nature.
\item A common-knowledge prior distribution 
  $\D$ on \emph{valuation matrices} $\V \in [0,1]^{n \times M}$
\end{itemize}


We assume the state of nature $\theta \in \Theta$, describing the item
  being sold in the auction, is first drawn from $\lambda$, and revealed
  to the auctioneer but not to the bidders. 
Then the auctioneer reveals a public signal $\sigma$, 
  a (partial) description of the item,
  according to some signaling scheme (See Section
  \ref{prelim:schemes}). 
Subsequently, players' valuations $\V \in
 [0,1]^{n \times M}$ are drawn from $\D$, 
 where $\V_{i\theta}$ is player $i$'s value for the 
 item described by $\theta$, and each
 player $i$ privately learns his valuation $\V_i$. 
Finally a second-price auction is run, where bidders bid according to their
 private valuations and their posterior belief (after learning
 $\sigma$) regarding the configuration of the item being
 sold.\footnote{As mentioned in \cite{DIR14}, the particular choice
 of auction is immaterial.
}

We note that we assume that $\D$ and $\lambda$ are independent. 
We also emphasize that the auctioneer knows nothing regarding 
  $\V$ besides its distribution $\D$ prior to running the auction, 
  and that the bidders know nothing regarding $\theta$ 
  besides its distribution $\lambda$ and the signal $\sigma$.

For our algorithmic results, we assume that $\lambda$ is represented
  explicitly as a vector, and that $\D$ is given explicitly as a list of
  $n \times M$ valuation matrices $\V^1,\ldots,\V^r \in [0,1]^{n\times M}$ with
  associated probabilities $\rho_1,\ldots,\rho_r$ summing to
  $1$. 
However, as we show in Section \ref{sec:auctionquasi}, our
  results also generalize to the case in which $\D$ can only be sampled
  efficiently. 
The latter includes, as a special case, the scenario in
  which players' valuations are independent of each other, with each
  represented explicitly as a list of vectors in $[0,1]^M$ with
  associated probabilities.

\subsection{Signaling Schemes} \label{prelim:schemes}

We examine policies whereby a principal reveals partial information regarding 
  the state of nature $\theta$ to the players. 
For our main results, we require that the principal reveal the same 
  information to all players in the game. 
A \emph{symmetric signaling scheme} is given by a set $\Sigma$ of 
  \emph{signals}, and a (possibly randomized) map $\varphi$ 
  from states of nature $\Theta$ to signals $\Sigma$. 
Abusing notation, we use $\varphi(\theta,\sigma)$ to denote 
  the probability of announcing signal $\sigma \in \Sigma$ 
  conditioned on the state of nature is $\theta \in \Theta$. 
We restrict attention to signaling schemes with a finite 
  set of signals $\Sigma$, and this is without loss of 
  generality when $\Theta$ is finite. 
We elaborate on this after describing the convex 
  decomposition interpretation of a signaling scheme. 


We note that signaling schemes are in one-to-one correspondence with
  \emph{convex decompositions} of the prior distribution $\lambda \in
  \Delta_M$ --- namely, distributions supported on the simplex
  $\Delta_M$, and having expectation $\lambda$.
Formally, a signaling scheme $\varphi: \Theta \to \Sigma$ 
  corresponds to the convex decomposition  
$\lambda = \sum_{\sigma \in \Sigma} \alpha_\sigma \cdot  \mu_\sigma,$
where $\alpha_\sigma= \Pr [\varphi(\theta) = \sigma] = \sum_{\theta
  \in \Theta} \lambda(\theta) \varphi(\theta,\sigma),$ and
$\mu_\sigma(\theta) = \Pr [ \theta | \varphi(\theta) = \sigma] =
\frac{\lambda(\theta) \varphi(\theta,\sigma)}{\alpha_\sigma}.$
Note that $\alpha_\sigma$ is the probability of announcing signal
$\sigma$, and $\mu_\sigma \in \Delta_M$ is the \emph{posterior distribution} 
  of $\theta$ conditioned on signal $\sigma$. The
  converse is also true: every convex decomposition of $\lambda$
  corresponds to a signaling scheme.

We judge the quality of a signaling scheme by the outcome it induces
  signal by signal. 
Specifically, the principal is equipped with an
  objective function of the form $\sum_\sigma \alpha_\sigma \cdot f(
  \mu_\sigma)$, where $f: \Delta_M \to \RR$ is some function mapping a
  posterior distribution to the quality of the equilibrium chosen by the
  players. 
For example, $f$ may be the social welfare at the induced
  equilibrium, or any weighted combination of players' utilities at
  equilibrium, or something else entirely. 
In this setup, one can show that there always exists an
  signaling scheme with a finite set of signals which maximizes our
  objective, so long as the states of nature are finitely many.  
The optimal choice of signaling scheme is related 
  to the \emph{concave envelope} 
  $f^+$ of the function $f$.\footnote{$f^+$ is the
  point-wise lowest concave function $h$ for which $h(x) \geq f(x)$
  for all $x$ in the domain. 
Equivalently, the hypograph of $f^+$ is
  the convex hull of the hypograph of $f$.}
Specifically, such a signaling scheme 
  achieves $\sum_\sigma \alpha_\sigma \cdot f( \mu_\sigma) = f^+(\lambda)$. 
Application of Caratheodory's theorem to the hypograph of $f$, 
  therefore, shows that $M+1$ signals suffice.


For concreteness, the reader can think of a signaling scheme $\varphi$
  as represented by the matrix of pairwise probabilities
  $\varphi(\theta,\sigma)$. 
Since we only consider games where the
  states of nature, and therefore also the number of signals w.l.o.g.,
are polynomially many in the description size of the game, this is a
compact representation. The representation of $\varphi$ as a convex
decomposition would do equally well, as both representations can be
efficiently computed from each other.



\subsection{Meshing Schemes at a High Level}

Note that direct discretization of the space of posterior beliefs $\Delta_M$
  is usually impractical when $M$ is large.
Instead, our meshing scheme overcomes this 
  ``curse of dimensionality'' 
   by using a partition of $\Delta_M$ into a quasi-polynomial number of 
   subspaces (of posterior beliefs).
We then select at most one posterior distribution from each component of the 
  partition as the convex decomposition of the prior $\lambda$.
Naturally, such partition of $\Delta_M$ needs to satisfy 
  both the mathematical property that there exists an approximately optimal 
  convex decomposition of $\lambda$  that respects the partition,
  and the algorithmic property that such convex decomposition
  can be computed efficiently.

For Bayesian games, we use a net of the equilibrium space
  known to exists due to Lipton, Markakis, and Mehta \cite{lmm03} to induce
  the desired partition of $\Delta_M$, and use linear programming
  to compute the convex decomposition of $\lambda$ with respect to the partition.
For Bayesian auctions, we prove that logarithmic sized subsets of items
  can be used to define the desired partition of $\Delta_M$
  that can be used in the submodular function maximization formulation
  developed in \cite{DIR14}.

\subsection{Equilibria and Objectives} \label{prelim:equilibria}


Given a Bayesian game, a symmetric signaling scheme $(\alpha, \mu)$ with signals $\Sigma$ induces $|\Sigma|$ subgames, one for each signal. The subgame corresponding to signal $\sigma \in \Sigma$ is played with probability $\alpha_\sigma$, and players' (common) beliefs regarding the state of nature in this subgame are given by the posterior distribution $\mu_\sigma \in \Delta_M$. The quality of a symmetric signaling scheme in such a game is contingent on a choice of an \emph{equilibrium  concept}, an \emph{equilibrium selection rule}, and an \emph{objective function}. 

\subsubsection*{Normal Form Games}
For normal form games, we adopt the  approximate Nash equilibrium as our equilibrium concept. There are two variants.

\begin{definition}
  Let $\eps \geq 0$. In an $n$-player $m$-action normal form game with expected payoffs in $[-1,1]$ given by  tensors $\A_1,\ldots,\A_n$, a mixed strategy profile $x_1,\ldots,x_n \in \Delta_m$ is an $\epsilon$-Nash Equilibrium ($\eps$-NE) if \[\A_i(x_1,\ldots,x_n) \geq \A_i(t_i, x_{-i}) - \eps \] for every player $i$ and alternative pure strategy $t_i \in [m]$.
\end{definition}

\begin{definition}
 Let $\eps \geq 0 $. In an $n$-player $m$-action normal form game with expected payoffs in $[-1,1]$ given by  tensors $\A_1,\ldots,\A_n$, a mixed strategy profile $x_1,\ldots,x_n \in \Delta_m$ is an $\epsilon$-well-supported Nash equilibrium ($\eps$-WSNE) if \[\A_i(s_i, x_{-i}) \geq \A_i(t_i, x_{-i}) - \eps\] for every player $i$, strategy $s_i$ in the support of $x_i$, and alternative pure strategy $t_i \in [m]$.
\end{definition}

Clearly, every $\epsilon$-WSNE is also an $\epsilon$-NE. Note that we omitted reference to the state of nature in the above definitions --- in a subgame corresponding to posterior beliefs $\mu \in \Delta_M$, we naturally use tensors $\A^\mu_1,\ldots\A^\mu_n$ instead.

Fixing an equilibrium concept ($\eps$-NE or $\eps$-WSNE), a Bayesian game $(\A,\lambda)$, and a signaling scheme $\varphi: \Theta \to \Delta(\Sigma)$, an \emph{equilibrium selection rule} distinguishes an equilibrium strategy profile $(x_1^\sigma,\ldots,x^\sigma_n)$ to be played in each subgame $\sigma$ --- we call the tuple $X=\set{x^\sigma_i: \sigma \in \Sigma, i \in [n]}$ a \emph{Bayesian equilibrium} of the game $(\A,\lambda)$ with signaling scheme $\varphi$. Together with the prior $\lambda$, the Bayesian equilibrium $X$ induces a distribution $\Gamma \in \Delta_{\Theta \cross [m]^n}$ over states of play --- we refer to $\Gamma$ as a \emph{distribution of play}. We say $\Gamma$ is  \emph{implemented} by signaling scheme $\varphi$ in $\eps$-NE ($\eps$-WSNE). This is analogous to implementation of allocation rules in traditional mechanism design.

Our results concern objectives which depend only on the state of play,
and we seek to maximize the objective in expectation over the
distribution of play. These include, but are not restricted to, the
social welfare of the players, as well as weighted combinations of
player utilities. Formally, our objective is given by a tensor $\F:
\Theta \cross [m]^n \to [-1,1] $.\footnote{Equivalently, we may think
  of the objective as the payoff tensor of an additional player in the
  game.} We seek a signaling scheme $\varphi: [M] \to \Sigma$, as well
as a Bayesian $\eps$-NE ($\eps$-WSNE) $X=\set{x^\sigma_i \in \Delta_m:
  i \in [n], \sigma \in \Sigma}$, maximizing $\Ex
[\F(\theta,\vec{s})]$ over the resulting distribution of play.
We use $OPT_{\textit{NE}}^\eps(\A,\lambda)$ and $OPT_{\textit{WSNE}}^\eps(\A,\lambda)$ to denote the optimal values of these two optimization problems. As shorthand, we use $f(\varphi,X) = \Ex_{\theta \sim \lambda} \Ex_{\sigma \sim \varphi(\theta)} \Ex_{\vec{s} \sim x^\sigma} [ \F(\theta, \vec{s})]$ to denote the expectation of our objective  $\F$ for signaling scheme $\varphi$ and corresponding Bayesian equilibrium $X$. 


\subsubsection*{Single Item Auctions}

In our second-price auction game, we adopt the (unique) dominant-strategy truth-telling equilibrium as our solution concept. Specifically, given a symmetric signaling scheme $\varphi: \Theta \to \Sigma$, for the subgame corresponding to the signal $\sigma \in \Sigma$  it is a dominant strategy for player $i$ to bid $\Ex_{\theta \sim \lambda } [ \V_{i\theta} | \varphi(\theta) = \sigma ]$ --- his posterior expected value for the item conditioned on the received signal~$\sigma$.

For our objective, we restrict attention to the \emph{social welfare} -- the expected value of the winning bidder for the item they win. Given a signaling scheme $\varphi$ and a valuation matrix $\V \in [0,1]^{n \times M}$, assuming bidders play the truth-telling equilibrium this is given by:
\begin{align*}
\textit{welfare}(\varphi,\lambda,\V) &=  \Ex_{\theta \sim \lambda} \Ex_{\sigma \sim \varphi(\theta)} \left[  \max_{i=1}^n \Ex_{\theta \sim \lambda } \left[ \V_{i\theta} | \varphi(\theta) = \sigma \right] \right]  \\
&=    \sum_{\sigma \in \Sigma} \max_{i=1}^n \sum_{\theta \in \Theta} \lambda(\theta) \varphi(\theta,\sigma) \V_{i\theta}
\end{align*}
When $\V$ is drawn from a prior distribution $\D$,
the expected welfare is given by 
\[\textit{welfare}(\varphi,\lambda,\D) = \Ex_{\V \sim \D}[ \textit{welfare}(\varphi,\lambda,\V)].\]

As noted in \cite{DIR14}, convexity of the welfare in the probabilities $\varphi(\theta,\sigma)$ implies that full-information-revelation is optimal, in general. However, we consider signaling subject to a communication constraint  --- i.e. with $\Sigma = \set{1,\ldots,k}$ for an input parameter $k < M$ limiting the number of different messages describing the item for sale.

\section{Signaling in Normal Form Games}
\label{sec:game}
We now consider signaling in explicitly represented games when the
adopted solution concept is the $\eps$-Nash equilibrium or the
$\eps$-well-supported Nash equilibrium. We prove the following
bi-criteria result.

\begin{theorem}\label{thm:gamemain}
  Fix $\eps > 0$, $\delta \geq 0$. Given as input an explicitly-described Bayesian normal form game $(\A,\lambda)$ with $n=O(1)$ players, $m$ actions, and $M$ states of nature, and an objective $\F:[M] \cross [m]^n \to [-1,1]$,  there is an algorithm with runtime $\poly(M, m^{{\ln m}/{\eps^2} })$ which outputs a signaling scheme $\varphi$ and corresponding Bayesian $(\epsilon+\delta)$-equilibrium $X$ satisfying $f(\varphi,X) \geq OPT^{\delta}(\A,\lambda) - \eps$. This holds for both approximate NE and approximate WSNE.
\end{theorem}

When the number of players is constant, we can in
quasi-polynomial time approximate the optimal reward from signaling
while losing an additive $\epsilon$ in the objective as well as in the
incentive constraints. Our proof of this theorem hinges on three main
lemmas: the first is drawn from the work of Lipton et al.~\cite{lmm03}
regarding the existence of a quasi-polynomial-sized ``net'' of the
space of equilibria; the second lemma states that the posterior
beliefs implementing a particular approximate equilibrium form a
simple polytope, in doing so reducing the signaling problem to
optimization over convex decompositions of $\lambda$ into a family of
posteriors, each belonging to a given polytope; and the third lemma
shows that optimization over such convex decompositions reduces to a
linear program.

\begin{lemma} [Lipton, Markakis and Mehta~\cite{lmm03}] \label{lem:lmm}
  Fix $n=O(1)$. For each integer $m$ and $\eps > 0$, there is a family of mixed strategies $\N_{m,\eps} \sse \Delta_m$ with $|\N_{m,\eps}| \leq m^{O({\log m}/{\eps^2})}$ such that for every $n$-player $m$-action normal form game $\A$, objective tensor $\F: [m]^n \to [-1,1]$, and $\delta$-NE  (WSNE) $x$ of $\A$, there is an $(\epsilon + \delta)$-NE (WSNE) $y$ of $\A$ such that each $y_i \in \N_{m,\eps}$,  
and $|\Ex_{s\sim x}[\F(s)] - \Ex_{s \sim y}[\F(s)]| \leq \eps$. Moreover, 
$\N_{m,\eps}$ can be enumerated in time $m^{O({\log m}/{\eps^2})}$.
\end{lemma}

\begin{lemma} \label{lem:posterior_polytope}
  Fix a normal form game of incomplete information $\A$ with $n$ players, $m$ actions, and $M$ states of nature. Consider a mixed strategy profile $x=(x_1,\ldots,x_n)$ with $x_i \in \Delta_m$. For each $\epsilon \geq 0$, the class of posterior beliefs inducing $x$ as an $\epsilon$-NE (WSNE) is a convex polytope described by $\poly(n,m)$ linear inequalities.
\end{lemma}

\begin{lemma} \label{lem:decompose}
Given a family of non-empty polytopes $\P_1,\ldots,\P_t \sse \Delta_M$ described by $\ell$ inequalities each, a point $\lambda \in \Delta_M$, and linear objectives $w_1,\ldots,w_t \in \RR^M$, the non-linear optimization problem \eqref{lp:decompose} can be solved in $\poly(t,\ell,M)$ time.
\end{lemma}

\begin{lp}\label{lp:decompose}
  \maxi{\sum_{\sigma=1}^t \alpha_\sigma w_\sigma \cdot  \mu_\sigma}
\st
\con{\alpha \in \Delta_t}
\con{\sum_{\sigma=1}^t \alpha_\sigma \mu_\sigma  = \lambda}
\qcon{\mu_\sigma \in \P_\sigma}{\sigma=1,\ldots,t}
\end{lp}

Before proving each of these lemmas, we first elaborate on how
they imply Theorem \ref{thm:gamemain}.  Given a signaling scheme
$\varphi$ with decomposition form $(\alpha,\mu)$, and an (approximate)
equilibrium $x^\sigma$ for each subgame corresponding to $\sigma$, the
objective value is
\[f(\varphi,x)  = \sum_{\sigma \in \Sigma} \alpha_\sigma \F(\mu_\sigma, x^\sigma)\] where $\F(\mu,x)$ denotes $\Ex_{\theta \sim \mu} \Ex_{\vec{s} \sim x} \F(\theta,\vec{s})$. 

Lemma \ref{lem:lmm} implies that, in order to complete the proof of
Theorem \ref{thm:gamemain}, it suffices to show how to exactly
optimize, in the claimed time, over signaling schemes in which
$x^\sigma \in \N_{m,\eps}^n$ for each signal $\sigma \in \Sigma$. We
may restrict attention to signaling scheme/equilibrium combinations
for which each mixed strategy profile $x \in \N_{m,\eps}^n$ is
selected for at most one subgame: when $x$ is the equilibrium for both
the subgames $\A^{\sigma_1}$ and $\A^{\sigma_2}$, we can ``merge'' the two
signals $\sigma_1$ and $\sigma_2$ into a signal $(\sigma_1,\sigma_2)$,
giving rise to a new subgame $\A^{(\sigma_1,\sigma_2)}$ with posterior
belief $\mu_{(\sigma_1,\sigma_2)} =
\frac{\alpha_{\sigma_1}}{\alpha_{\sigma_1} + \alpha_{\sigma_2}} \mu_{\sigma_1}
+ \frac{\alpha_{\sigma_2}}{\alpha_{\sigma_1} + \alpha_{\sigma_2}}
\mu_{\sigma_2}$ and probability $\alpha_{(\sigma_1,\sigma_2)} =
\alpha_{\sigma_1} + \alpha_{\sigma_2}$.
Lemma \ref{lem:posterior_polytope} implies that $x$ remains an (approximate) equilibrium of the merged subgame. Moreover, the objective is unchanged because $\F(\mu,x)$ is linear in its first argument. 

For notational convenience we assume that each $x \in
\N_{m,\eps}^n$ is induced as an equilibrium of exactly one subgame, by
allowing signals which occur with probability $0$, and discarding
strategy profiles in $\N_{m,\eps}^n$ which can not be induced as
equilibria of any posterior belief.
The latter can be done in
  polynomial time, by checking whether the corresponding polytope (as
  given by Lemma \ref{lem:posterior_polytope}) is empty. 
Writing
$\N_{m,\eps}^n = \set{x^1,\ldots,x^t}$ for $t=m^{O(\log m/{\eps^2})}$, our optimization task can be written as follows.

\begin{lp}\label{lp:decompose2}
  \maxi{\sum_{\sigma=1}^t \alpha_\sigma \F(\mu_\sigma,x^\sigma)}
\st
\con{\alpha \in \Delta_t}
\con{\sum_{\sigma=1}^t \alpha_\sigma \mu_\sigma = \lambda}
\qcon{\mbox{$x^\sigma$ is an equilibrium of $\A^{\mu_\sigma}$}}{\sigma=1,\ldots,t}
\end{lp}

Lemma \ref{lem:posterior_polytope}, and the linearity of $\F(.,.)$ in
its first argument, imply that 
optimization problem
\eqref{lp:decompose2} is of the form given in \eqref{lp:decompose}
with $\ell = \poly(n,m)$. Lemma \ref{lem:decompose}, and our
assumption that $n$ is a constant, imply that \eqref{lp:decompose2}
can be solved in time $\poly(M,m^{({\ln m}/{\eps^2})})$.
This completes the proof of Theorem \ref{thm:gamemain}.


\subsection{Proof Sketch of Lemma \ref{lem:lmm}}
The special case of this lemma for $\epsilon$-NE and $\delta=0$ follows directly from the statement \cite[Theorem 2]{lmm03}, by including an additional player in the game with no actions, and payoffs given by the objective tensor $\F$ evaluated on the strategies of the other (real) players. The set $\N_{\eps,m}$ is taken to be the family of all mixed strategies which are uniformly distributed on a multiset contained in $[m]$ of size $\frac{3(n+1)^2 \ln (n+1)^2 m}{\eps^2}$.

More generally, at the heart of \cite[Theorem 2]{lmm03} is the fact that, for every mixed strategy profile $x \in \Delta_m^n$, one can choose a mixed strategy profile $y \in \N_{\eps,m}^n$, with $supp(y_i) \sse supp(x_i)$ for each $i \in [n]$, so that $|\A_i(j,y_{-i}) - \A_i(j,x_{-i})| \leq \frac{\eps}{2}$ for every player $i$ and pure strategy $j \in [m]$. This has the implication that $y$ is an $(\epsilon+\delta)$-NE (WSNE) when $x$ is a $\delta$-NE (WSNE). Moreover, when the objective is included as a player in the game with no nontrivial strategies (which is without loss of generality), this also implies that $|\F(y) - \F(x)| \leq \eps$.

\subsection{Proof of Lemma \ref{lem:posterior_polytope}}
For $x$ to be an $\eps$-NE of $\A^\mu=\sum_{\theta=1}^M \mu(\theta) \A^\theta$, for $\mu \in \Delta_M$, the following must hold:
\begin{lp}
  \qcon{\sum_{\theta =1}^M \mu(\theta) \A_i^\theta(x) \geq \sum_{\theta =1}^M \mu(\theta) \A^\theta_i(j, x_{-i}) - \eps}{i \in [n], j \in [m]}
\end{lp}
 For an $\eps$-WSNE, the analogous system of inequalities is:
\begin{lp}
  \qcon{\sum_{\theta =1}^M \mu(\theta) \A_i^\theta(j,x_{-i}) \geq \sum_{\theta =1}^M \mu(\theta) \A^\theta_i(k, x_{-i}) - \eps}{i \in [n],j \in supp(x_i),  k \in [m]}
\end{lp}
Since $x$ is fixed, in both cases we have a system $\poly(n,m)$ linear inequalities in $\mu$.

\subsection{Proof of Lemma \ref{lem:decompose}}
We write an equivalent linear program via a change of variables. Specifically, we  let $\gamma_\sigma= \alpha_\sigma \mu_\sigma$.
Observe that after this change \eqref{lp:decompose} becomes:

\begin{lp}\label{lp:decompose3}
  \maxi{\sum_{\sigma=1}^t  w_\sigma \cdot \gamma_\sigma}
\st
\con{\alpha \in \Delta_t}
\con{\sum_{\sigma=1}^t  \gamma_\sigma = \lambda}
\qcon{\frac{\gamma_\sigma}{\alpha_\sigma} \in \P_\sigma}{\sigma=1,\ldots,t}
\end{lp}

\eqref{lp:decompose3} is not yet a linear program. However, note that
the constraint $\alpha \in \Delta_t$ is implied if we simply add the
constraints $\gamma_\sigma \succeq 0$. Moreover, because
${\gamma_\sigma}/{\alpha_\sigma} \in \Delta_M$, $\alpha_\sigma =
\sum_{\theta} \gamma_\sigma(\theta)$ holds for every feasible
solution, allowing us to simplify the constraint
${\gamma_\sigma}/{\alpha_\sigma} \in \P_\sigma$. 
Since
$\P_\sigma$ is described an explicit linear system $A^\sigma y \preceq
b^\sigma$, the non-linear system of inequalities $A^\sigma
{\gamma_\sigma}/{\alpha_\sigma} \preceq b^\sigma$ can be
re-written as the linear system $A^\sigma \gamma_\sigma \preceq
(\sum_{\theta} \gamma_\sigma(\theta)) b^\sigma$. This results in an
equivalent linear program with variables $\gamma_1,\ldots,\gamma_t \in
\RRp^M$, from which $\alpha_\sigma = \sum_{\theta}
\gamma_\sigma(\theta)$ and $\mu_\sigma = \gamma_\sigma /
(\sum_{\theta} \gamma_\sigma(\theta))$ can be recovered efficiently.


\subsection{Remarks}
\paragraph
{Zero-sum games}
When applied to two-player zero-sum games with the objective to maximize
  the first-player's payoff, our signaling scheme provides a
  stronger guarantee.
In such setting, both players retain the same payoff in any exact Nash equilibrium.
Also, any $\eps$-equilibria gives a payoff that is $\eps$
  close to the playoff of any exact equilibrium.  
Thus, the signaling scheme provided in Theorem~\ref{thm:gamemain} 
can be directly compare with the quality of the optimal signaling scheme
  without concerning equilibrium selection,
instead of a bi-criteria guarantee.

\paragraph
{Reducing the number of signals}
Although the signaling scheme provided in Theorem~\ref{thm:gamemain}
  might use quasi-polynomial number of signals, we can reduce
  the number of signals to $M+1$.
Let $w_\lambda$ be the objective value of the signaling scheme,
  and consider the set of $t$ signals used $\mu_1 \cdots \mu_t$
  and their corresponding expected objective values $w_1 \cdots w_t$.
Observe that the $M+1$ dimension point $(w_\lambda,\lambda)$ is
  a convex combination of the set of points
  $P=\{ (w_1,\mu_1) \cdots (w_t,\mu_t) \}$.
Since $w_\lambda$ is maximized, $(w_\lambda,\lambda)$ belongs to
  some facet of the convex hull of $P$.
Hence by Carath\'{e}odory's theorem, $(w_\lambda,\lambda)$ can be written
  as a convex combination of only $M+1$ points from $P$, and such decomposition
  can be computed in time polynomial in the size of $P$.
This decomposition gives a valid signaling scheme with the same
  objective value, using only $M+1$ signals.

  

\paragraph{Stackelberg games}
Our result can be extended to {\em Stackelberg games}
  which often arise in security games \cite{D14}.
Recall that in a Stackelberg game \cite{stackelberg}, one player (the leader) 
  first commits to a (mixed) strategy, 
  and then all other players (followers)  simultaneously play their strategies
  upon learning the leader's strategy.
Our result can be readily extended to Bayesian Stackelberg games
  when the objective of the signaling scheme is to maximize leader's payoff.
In this case, we can simply drop the constraints 
  regarding the leader in the polytopes defined in 
  Lemma \ref{lem:posterior_polytope},
  and only require the followers to play an approximate equilibrium
 in our algorithm presented in Theorem \ref{thm:gamemain}.

\section{Signaling in Bayesian Single Item Auctions}
\label{sec:auction}

In this section, we consider the signaling problem in a Bayesian single item auction, as described in Section \ref{sec:prelims}. As in \cite{DIR14}, our goal is to maximize the social welfare subject to a communication constraint  $k$ on  signaling scheme. We present a quasi-polynomial time approximation scheme for this problem when the valuation distribution $\D$ is given explicitly, and as a corollary also when $\D$ is given by a sampling oracle. To aid in our proof, we begin with some technical background drawn from related work. In the discussion of this section, we will use 
  {\em $(n,M,k,r)$-second-price signaling} to denote the welfare-maximization signaling problem with $n$ bidders, $M$ configurations of the item, $k$ signals, and an explicitly-described valuation distribution supported on $r$ matrices.





\subsection{Background: Reducing to Submodular Maximization for small $r$}
\label{sec:auctionbg}
Recall that a signaling scheme is
a randomized map from states of nature to
  signals. A deterministic signaling scheme with $k$ signals, naturally, is a deterministic map $\varphi: \Theta \to [k]$. Equivalently, such a scheme corresponds  to a partition of the states of nature $\Theta$ into $k$ classes $\Theta_1,\ldots,\Theta_k$, one per signal. As might be apparent from Section \ref{sec:game}, in general games randomized signaling schemes may outperform their deterministic counterparts. However, this is not the case in our auction setting. 

\begin{lemma}[\cite{DIR14}]
For communication-constrained signaling in Bayesian second-price auctions, there always exists a deterministic signaling scheme which maximizes expected social welfare.
\end{lemma}
Thus instead of solving the signaling problem by assigning
probabilities in a continuous domain, we can exploit the fact that our
signaling problem has a combinatorial solution. This is the basis for
the $(1-1/e)$-approximation algorithm of \cite{DIR14} for
$(n,M,r,k)$-second price signaling, which we outline next.
This algorithm reduces the signaling problem to submodular optimization,\footnote{Recall that 
  a function $f: 2^\Omega \rightarrow \mathbb{R}$ on a ground set $\Omega$
  is {\em submodular} if it satisfies 
$f(X) + f(Y) \geq f(X \union Y) + f(X \intersect Y)$ for every $X,Y \sse \Omega$.}
and is computationally efficient when $r$, the support size of the valuation distribution $\D$, is small.
Using $\V^1,\ldots,\V^r$ and $\rho_1,\ldots,\rho_r$ to denote the support of $\D$ and the
corresponding probabilities, respectively, consider a deterministic
signaling scheme in partition form $\Theta_1,\ldots,\Theta_k$. Such a
signaling scheme induces $kr$ subgames, one for each pair
$(\Theta_\sigma,\V^t)$. In each such subgame, there is a unique
winning player of the auction --- the player $i$ with maximum
posterior value \[\frac{1}{\lambda(\Theta_\sigma)} \sum_{\theta \in \Theta_\sigma} \lambda(\theta) \V^t(i,\theta).\]
The social welfare, therefore, is given by
\begin{equation}
  \label{eq:auction_welfare}
  \sum_{t=1}^r \sum_{\sigma=1}^k \max_{i\in [n]}  \sum_{\theta \in \Theta_\sigma} \rho_t \lambda(\theta) \V^t(i,\theta).
\end{equation}

The key observation behind the algorithm of \cite{DIR14} is that, given a ``guess'' for the winner of the auction in each of these $kr$ subgames, the optimal signaling scheme $(\Theta_1,\ldots,\Theta_k)$ can be recovered in $\poly(n,M,r)$ time. More generally, given a list $w_1,\ldots,w_k$, where $w_\sigma \in [n]^r$ is the \emph{winner tuple} for signal $\sigma$, we can write the expected social welfare of a deterministic signaling scheme $\Theta_1,\ldots,\Theta_k$ under the assumption that  $w_\sigma(t)$ wins in the event $(\sigma,t)$.
\begin{equation}
  \label{eq:auction_welfare_guesses}
\textit{welfare}(\Theta_1,\ldots,\Theta_k, w_1,\ldots,w_k) = \sum_{t=1}^r \sum_{\sigma=1}^k \sum_{\theta \in \Theta_\sigma} \rho_t \lambda(\theta) \V^t(w_\sigma(t)),\theta).
\end{equation}
It is now clear that the optimal scheme, assuming winner tuples $w_1,\ldots,w_k$, maps $\theta$ to the signal $\sigma$ maximizing $\sum_{t=1}^r \sum_{\sigma=1}^k  \rho_t \lambda(\theta) \V^t(w_\sigma(t), \theta).$ The resulting scheme has social welfare at least as given by the following expression, with equality holding for the optimal set of guesses $w$.

\begin{equation}
  \label{eq:welfare_tuples}
 \textit{welfare}(w_1,\ldots,w_k) = \sum_{\theta \in \Theta} \max_{\sigma \in [k]} \sum_{t=1}^r \rho_t \lambda(\theta) \V^t(w_\sigma(t)),\theta).
\end{equation}
We summarize this discussion by the following lemma. 
\begin{lemma}[\cite{DIR14}]\label{lem:recover_from_tuple}
  Consider an instance $(\lambda,\D)$ of $(n,M,r,k)$-second-price signaling. Given a list $w_1,\ldots,w_k \in [n]^r$ of winner tuples, a deterministic signaling scheme $\varphi: \Theta \to [k]$ with $\textit{welfare}(\varphi,\D,\lambda) \geq \textit{welfare}(w_1,\ldots,w_k)$ can be computed in $\poly(n,M,r)$ time. Moreover, there is a set of winner tuples $w^*_1,\ldots,w^*_k$ with $\textit{welfare}(w^*_1,\ldots,w^*_k)$ equal to the optimum welfare of a $k$-signal scheme.
\end{lemma}

The above lemma naturally leads to a set function maximization problem subject to a cardinality constraint of $k$.  By symmetry, we can think of a list of winner tuples $w_1,\ldots,w_k$ equivalently as an (unordered) set $W= \set{w_1,\ldots,w_k}$, assigning an order arbitrarily. Assuming winning tuples $W$, the signaling scheme of Lemma \ref{lem:recover_from_tuple} has welfare at least
\begin{equation}
  \label{eq:welfare_setfunction}
 \textit{welfare}(W) = \sum_{\theta \in \Theta} \max_{w \in W} \sum_{t=1}^r \rho_t \lambda(\theta) \V^t(w(t),\theta),
\end{equation}
with equality holding at optimality. The above function naturally extends to all  $W \sse [n]^r$, and is monotone nondecreasing and submodular. To see this, note that it can be written in the form $\sum_{\theta} \max_{w \in W} g_\theta(w)$ --- the function $\max_{w \in W} g_\theta(w)$ is monotone and submodular for each fixed $\theta$, and therefore the sum of such functions is as well. Thus maximizing the social welfare subject to communication 
  constraint $k$ can be formulated as the following submodular maximization problem subject to a cardinality constraint:
\begin{eqnarray}\label{eqn:submodular}
\max_{|W|=k, W \subseteq [n]^r} \textit{welfare}(W).
\end{eqnarray}
The problem of maximizing a monotone
  submodular function subject to cardinality constraint is NP-hard in general, though
  can be approximated to within to a factor of $(1-1/e)$ using a simple greedy algorithm~\cite{nwf78}. The runtime of the algorithm is polynomial in the size of the ground set, as well as the time needed to evaluate the function. This leads to a deterministic $(1-1/e)$-approximation algorithm for $(n,M,k,r)$-second-price auction signaling which runs in time $\poly(n^r,M)$.

\begin{theorem}[\cite{DIR14}]
\label{thm:dir-main}
There is a deterministic $\poly(n^r,M)$-time $\left(1-\frac{1}{e}\right)$-approximation algorithm for welfare maximization in $(n,M,k,r)$-second-price auction signaling.
\end{theorem}

Finally, we note that the above approximation guarantee is optimal,
assuming $P \neq NP$, even when $r=1$ \cite{DIR14}.
Thus, 
our quasi-polynomial-time algorithms necessarily must yield no better than
a $(1-1/e)$-approximation, unless NP-complete problems can be solved
in quasi-polynomial time.

\subsection{Quasi-Polynomial Algorithm for Auction Signaling}
\label{sec:auctionquasi}
The exponential dependence on $r$ in Theorem \ref{thm:dir-main} limits its applicability. In this section, we present an approximation algorithm for the $(M,n,r,k)$-signaling problem with runtime exponential in $\log (n r)$ instead --- quasi-polynomial in the size of the explicit input representation. The algorithm achieves the same multiplicative approximation guarantee of $(1-1/e)$, modulo an additive loss that can be made arbitrarily small. As a consequence, we get a randomized quasi-polynomial-time algorithm with a similar guarantee when valuation distributions are given by a sampling oracle. 

\begin{theorem}
\label{thm:auction}
Given an instance of $(n,M,k,r)$-second-price signaling in the explicit model, and an approximation parameter $\epsilon > 0$, there is a deterministic algorithm which runs in $M^{O({\log(nr)}/{\epsilon^2})}$-time, and computes a signaling scheme with $k$ signals and expected social welfare $(1-1/e)(OPT-\epsilon)$, where $OPT$ denotes the optimal welfare of a $k$-signal scheme.
\end{theorem}

\begin{corollary}
\label{cor:auction}
Given an instance of $(n,M,k,r)$-second-price signaling where the valuation distribution is given via a sampling oracle, and parameters $\epsilon,\delta > 0$, there is a Monte Carlo algorithm which runs in time $M^{O({\log(n M \delta^{-1} \epsilon^{-2})}/{\eps^2})}$, and with success probability $1-\delta$ outputs a signaling scheme with $k$ signals and expected social welfare $(1-1/e)(OPT-2\epsilon)$. 
\end{corollary}

Corollary \ref{cor:auction} follows from Theorem \ref{thm:auction} by using standard Monte Carlo sampling and convergence arguments, exploiting the fact that there are at most $k^M$ deterministic signaling schemes. We include a self contained proof in Section \ref{subsec:auctioncor}.
Theorem \ref{thm:auction}, on the other hand, follows from the following Lemma,
which we will prove in Section \ref{pf:auctionlemma}.

\begin{lemma}\label{lem:auction}
Given an instance of the  $(n,M,k,r)$-second-price signaling problem,  and a parameter $\epsilon>0$,  there is a family of winner tuples  $\N_\epsilon \subseteq [n]^r$ with $|\N_\epsilon| = M^{O({\log(nr)}/{\epsilon^2})}$, satisfying 
\begin{eqnarray}\label{eqn:approxSubmodular}
\max_{|W|=k, W \subseteq {\cal N}_\epsilon} \textit{welfare}(W) 
 \ge \max_{|W|=k, W \subseteq [n]^r} \textit{welfare}(W) - \epsilon.
\end{eqnarray}
Moreover, $\N_\epsilon$ can be enumerated in time $M^{O({\log(nr)}/{\epsilon^2})}$.
\end{lemma}

Armed with Lemma \ref{lem:auction}, we complete the proof of Theorem \ref{thm:auction}. We follow the blueprint outlined in Section \ref{sec:auctionbg}, using $\N_\epsilon$ as our ground set of winner tuples in lieu of  $[n]^r$. Specifically, we use the greedy algorithm to approximately maximize $\textit{welfare}(W)$ over $W \sse \N_\eps$ with $|W|=k$. This leads to a signaling scheme with welfare at least $(1-1/e)(OPT-\eps)$. Given the reduced size of the ground set of our submodular function, the runtime is now $\poly(|\N_\eps|,n,M,r) = M^{O({\log(nr)}/{\epsilon^2})}$.

\subsection{Proof of Lemma \ref{lem:auction}}\label{pf:auctionlemma}
Given a multiset of item
configurations $Y \sse \Theta$, we distinguish the winner tuple
\emph{induced} by $Y$, which simply lists the winning player for each
$V^t \in supp(\D)$ assuming $\theta$ is drawn uniformly from $Y$.
Formally, the winner tuple $w(Y)=(w_1(Y),\ldots,w_r(Y))$ induced by
$Y$ is defined by $w_t(Y) = \argmax_{i=1}^n \Ex_{\theta \sim Y}[ \V^t(i,\theta)]$.  We then define our set $\N_\eps$ as follows.
\[\N_\eps = \set{ w(Y) : |Y| = \frac{2\ln(4nr)}{\epsilon^2}, Y \sse \Theta}\]

$\N_\eps$ is of the claimed size. It remains
to show that considering only winner tuples in $\N_\eps$ in the
optimization problem $\max_{|W| = k} \textit{welfare}(W)$ results in a
loss of welfare no more than an additive $\epsilon$.  This hinges on
the following claim.

\begin{claim}
  Fix $\eps>0$. For each $X \sse \Theta$, there is multiset $Y \sse X$ with $|Y|= {2\ln(4nr)}/{\epsilon^2}$ such that,
\[\left|\Ex_{\theta \sim \lambda} [\V^t(i,\theta) | \theta \in X ] - \Ex_{\theta \sim Y} [\V^t(i,\theta)] \right| \leq \eps/2\]
for every $t=1,\ldots,r$ and $i =1,\ldots,n$.
\end{claim}
\begin{proof}
  Consider taking ${2\ln(4nr)}/{\epsilon^2}$ samples $Y$ from the distribution $\lambda | X$. Standard application of the Hoeffding bound and the union bound shows that $Y$ satisfies the desired guarantee with probability at least ${1}/{2} > 0$. The existential result then follows.
\end{proof}

As a consequence, for every deterministic signaling scheme
$(\Theta_1,\ldots,\Theta_k)$, we can choose a tuple of multisets
$(Y_1,\ldots,Y_k)$ with $Y_\sigma \sse \Theta_\sigma$ and $|Y_\sigma|
= {2\ln(4nr)}/{\epsilon^2}$ for each $\sigma \in [k]$, so that the
empirical distribution $Y_\sigma$ always approximates, to within an
additive $\eps/2$, the posterior expected value of every player in
subgame $\sigma$. Hence, using the empirical distribution $Y_\sigma$
to determine the winner of each signal $\Theta_\sigma$ preserves the
social welfare in the corresponding subgame up to an additive
$\epsilon$. When $(\Theta_1,\ldots,\Theta_k)$ is an optimal scheme,
this implies that $\textit{welfare}(w(Y_1),\ldots,w(Y_k)) \geq OPT -
\eps$, as needed.

\subsection{Proof of Corollary \ref{cor:auction}}
\label{subsec:auctioncor}

We fix an arbitrary deterministic signaling scheme $\varphi$ (i.e. partition on items), and consider the randomness in valuation matrices.
Let $f(\varphi,V)$ be the expected welfare induced by $\varphi$ when the valuation matrix is $V$, and let $f(\varphi,{\cal D})$ be the expected welfare induced by $\varphi$ when the valuation is drawn from ${\cal D}$.
By definition, $f(\varphi,{\cal D}) = E_{V \sim {\cal D}}[f(\varphi,V)] \in [0,1]$.
In order to estimate $f(\varphi,{\cal D})$, we take $r$ samples $V_1,\ldots,V_r$ from ${\cal D}$ and apply our algorithm in Theorem \ref{thm:auction} as if the valuation distribution is uniform on $V_1,\ldots,V_r$.
We apply Hoeffding's inequality to bound the probability that welfare estimated from $V_1,\ldots,V_t$ is far from its expectation.
\[
Pr\left(\left|\frac{1}{r} \sum_{j=1}^r f(\varphi,V_j) - f(\varphi,{\cal D})\right| \ge \epsilon / 2\right) \le 2e^{-r \epsilon^2/2},
\]
so by setting $r = \Theta( (M \log k + \log(\delta^{-1})) / \epsilon^2)$, we get this probability to be at most $\delta k^{-M}$.
Because there are at most $k^M$ deterministic signaling schemes, by union bound, we preserve the welfare up to an additive error $\epsilon/2$ for all of them. We conclude that, in this situation, our signaling scheme computed as in Theorem \ref{thm:auction} has expected welfare at least $(1-1/e)(OPT-2\epsilon)$.


\bibliographystyle{alpha}
\bibliography{agt,signal}
\end{document}